\documentclass[a4paper,UKenglish,cleveref, autoref]{lipics-v2019}
\usepackage{amsmath}
\usepackage{amsthm}

\hideLIPIcs
\nolinenumbers

\newcommand{\enc}[1]{{\mathcal{E}}({#1})}
\newcommand{\doot}[2]{\left\langle #1, #2\right\rangle}
\newcommand{\dec}[1]{{\mathcal{D}}({#1})}
\newcommand{\R}{\mathbb{R}}
\newcommand{\N}{\mathbb{N}}
\renewcommand{\S}{\mathbb{S}}

\renewcommand{\epsilon}{\varepsilon}

\newcommand{\Osymbol}{{O}}
\newcommand{\BO}[1]{\Osymbol\left(#1\right)}

\newcommand{\BOx}[1]{\Osymbol(#1)}

\title{The space complexity of inner product filters}

\author{Rasmus Pagh}{IT University of Copenhagen, Denmark \and BARC, Copenhagen, Denmark}{pagh@itu.dk}{https://orcid.org/0000-0002-1516-9306
}{Supported by grant 16582, Basic Algorithms Research Copenhagen (BARC), from the VILLUM Foundation.}
\author{Johan Sivertsen}{IT University of Copenhagen, Denmark}{johanvts@gmail.com}{https://orcid.org/0000-0002-3226-4209}{}
\authorrunning{R. Pagh and J. Sivertsen}
\Copyright{R. Pagh and J. Sivertsen}



\funding{The work leading to this research has received funding from the European Research Council under the European Union’s 7th Framework Programme (FP7/2007-2013) / ERC grant agreement no.~614331.}

\acknowledgements{We would like to thank Thomas Dybdahl Ahle for suggesting a simple proof for an inequality in section~\ref{sec:lower}, and the anonymous reviewers for many constructive suggestions, improving the presentation. Part of this work was done while the first author was visiting Google Research.}

\EventEditors{Carsten Lutz and Jean Christoph Jung}
\EventNoEds{2}
\EventLongTitle{23rd International Conference on Database Theory (ICDT 2020)}
\EventShortTitle{ICDT 2020}
\EventAcronym{ICDT}
\EventYear{2020}
\EventDate{March 30--April 2, 2020}
\EventLocation{Copenhagen, Denmark}
\EventLogo{}
\SeriesVolume{155}
\ArticleNo{16}

\begin{document}

\maketitle

\begin{abstract}
Motivated by the problem of filtering candidate pairs in inner product similarity joins we study the following \emph{inner product estimation} problem:
  Given parameters $d\in\N$, $\alpha>\beta\geq 0$ and unit vectors $x,y\in \R^{d}$ consider the task of distinguishing between the cases $\langle x, y\rangle\leq\beta$ and $\langle x, y\rangle\geq \alpha$ where $\langle x, y\rangle = \sum_{i=1}^d x_i y_i$ is the inner product of vectors $x$ and $y$.
	The goal is to distinguish these cases based on information on each vector encoded independently in a bit string of the shortest length possible.
	In contrast to much work on compressing vectors using randomized dimensionality reduction, we seek to solve the problem \emph{deterministically}, with no probability of error.
  Inner product estimation can be solved in general via estimating $\langle x, y\rangle$  with an additive error bounded by $\varepsilon = \alpha - \beta$.
  We show that $d \log_2 \left(\tfrac{\sqrt{1-\beta}}{\varepsilon}\right) \pm \Theta(d)$ bits of information about each vector is necessary and sufficient.
  Our upper bound is constructive and improves a known upper bound of $d \log_2(1/\varepsilon) + O(d)$ by up to a factor of~2 when $\beta$ is close to~$1$.
  The lower bound holds even in a stronger model where one of the vectors is known exactly, and an arbitrary estimation function is allowed.
\end{abstract}

\section{Introduction}

Modern data sets increasingly consist of noisy or incomplete information, which means that traditional ways of matching database records often fall short.
One approach to dealing with this in database systems is to provide \emph{similarity join} operators that find pairs of tuples satisfying a similarity predicate.
We refer to the book of Augsten and B{\"o}hlen~\cite{DBLP:series/synthesis/2013Augsten} for a survey of similarity joins in relational database systems, and to~\cite{conf/icdt/Pagh15} for an overview of theoretical results in the area.
Note that joins can be implemented using similarity search \emph{indexes} that allow searching a relation for tuples that satisfy a similarity predicate with a given query~$q$.
Thus we include works on similarity search indexes in discussion of previous work on similarity join.

In this paper we consider \emph{inner product similarity} predicates of the form $\langle x, y\rangle /geq \alpha$, where $x,y\in\R^{d}$ are real-valued vectors, i.e., the predicate is true for vectors whose inner product $\sum_{i=1}^d x_i y_i$ exceeds a user-specified threshold~$\alpha$.
This notion of similarity is important, for example, in neural network training and inference, (see~\cite{DBLP:conf/kdd/SpringS17}).
In that context, an inner product similarity join can be used for classifying a collection of vectors according to a neural network model.

Inner product similarity join is a special case of similarity join under Euclidean distance, implemented for example in the Apache Spark framework\footnote{\url{https://spark.apache.org/}}.
Conversely, it generalizes \emph{cosine similarity}, which has been studied for more than a decade (see, e.g.,~the influential papers~\cite{bayardo2007scaling, DBLP:conf/stoc/Charikar02} and more recent works such as~\cite{anastasiu2014l2ap,DBLP:conf/icassp/BaluFJ14,DBLP:conf/icde/ChristianiPS18,satuluri2012bayesian}). In recent years the general inner product similarity join problem has attracted increasing attention (see e.g.~\cite{ahle2016complexity,DBLP:conf/kdd/RamG12,DBLP:conf/nips/Shrivastava014,teflioudi2017exact,DBLP:conf/nips/YanLDCC18}).
Recently proposed practical inner product similarity join algorithms work by reducing the general problem to a number of instances with unit length vectors, which is equivalent to join under cosine similarity~\cite{DBLP:conf/nips/YanLDCC18}.

\medskip

{\bf Candidate generation approach.}
State-of-the-art algorithms computing similarity joins on high-dimensional vectors use a two-phase approach:
\begin{enumerate}
\item Generate a set of \emph{candidate pairs} $(x,y)$ that contains all pairs satisfying the predicate (keeping track of the corresponding tuples in the relations).
\item Iterate over the candidate pairs to \emph{check} which ones satisfy the predicate.
\end{enumerate}
Suppose for simplicity that both relations in the similarity join contain~$n$ tuples.
A na\"ive candidate generation phase would output all $n^2$ corresponding pairs of vectors.
For many data sets it is possible to reduce the number of candidate pairs significantly below $n^2$, but the check phase remains a bottleneck.
A direct implementation of the check phase would require full information about the vectors $(x,y)$, in practice $d$ floating point numbers per vector.
Though the inner product computation is trivial, for high-dimensional vectors the cost of transferring data from memory can be a bottleneck.

\medskip

{\bf Filtering candidate pairs using approximation.}
An approach to reducing communication is to \emph{approximate} inner products, which is enough to handle those candidate pairs that do not have inner product close to the threshold $\alpha$.
The exact inner product is computed only for the remaining pairs, often a small fraction of the set of all candidates.
We stress that globally, the join computation we consider is not approximate, but approximations are used to speed up parts of the algorithm.
(Note that under common assumptions in fine-grained complexity, the inner product similarity join problem is difficult in the worst case, even with approximation~\cite{DBLP:conf/focs/AbboudRW17,ahle2016complexity}.)

Such additional filtering of candidate pairs has been successfully used in ``Monte Carlo'' style randomized algorithms that allow the algorithm to sometimes fail to identify a pair satisfying the predicate, e.g.~\cite{satuluri2012bayesian,DBLP:conf/esa/0001CPV19}.
While the error in Monte Carlo algorithms can usually be made very small at a reasonable computation cost, such algorithms are not suitable in all settings. For example:
\begin{itemize}
\item Firm guarantees may be needed to comply with regulation, or to ensure a clear and consistent semantics of a system (such as a DBMS) in which the similarity estimation algorithm is part.
\item Guarantees on accuracy are shown under the assumption that the input data is independent of the random choices made by the algorithm. Technically this assumption may not hold if output from the algorithm can affect future inputs. Maybe more seriously, if vectors can be chosen adversarially based on (partial) knowledge of the randomness of the algorithm, for example obtained by timing attacks, the guarantees cease to hold (see e.g.~\cite{DBLP:conf/ccs/ClaytonPS19} for more discussion of adversarial settings).
\end{itemize}

In this paper we study what kind of approximation is possible without randomization, targeting settings where false negatives are not permitted, or where we cannot ensure that inputs are independent of any randomization used by the algorithm.

We seek to efficiently eliminate all candidate pairs that have inner product less than~$\beta$, for some $\beta$ smaller than the threshold $\alpha$, so that the number of remaining candidate pairs (for which an expensive inner product computation must be done) may be significantly reduced.
In order to not eliminate any candidate pair passing the threshold it is necessary and sufficient that the approximation is strong enough to distinguish the cases $\langle x, y\rangle\leq\beta$ and $\langle x, y\rangle\geq \alpha$.

\medskip

{\bf Similarity join memory bottlenecks.}
The complexity of similarity joins in the I/O model was studied in~\cite{DBLP:conf/esa/PaghPSS15}, which assumes that a block transfer moves $B$ vectors from or to external memory, and that internal memory can hold $M$ vectors.
Reducing the amount of data that needs to be transferred to evaluate a similarity predicate leads to a larger capacity of blocks as well as internal memory, in turn leading to a reduction in I/O complexity that is roughly proportional to the reduction in size.
The exact improvement is a bit more complicated because additional I/Os are needed to evaluate the exact inner products of pairs with similarity above $\beta$.
McCauley and Silvestri~\cite{mccauley2018adaptive} studied the related problem of similarity joins in MapReduce where considerations similar to the I/O model can be made.

\subsection{Our results}

Without loss of generality we can consider \emph{unit vectors}, since the general estimation problem can be reduced to this case by storing an (approximate) norm of each vector in space independent of the number of dimensions. Similarly, lower bounds shown for unit vectors imply lower bounds for arbitrary vector lengths by a scaling argument

We study the following version of the inner product estimation problem for unit vectors: Distinguish inner products smaller than $\beta$ from inner products larger than $\alpha$, for threshold parameters $\alpha$ and $\beta$.
This problem can of course be solved by estimating the inner product with additive error less than $\alpha - \beta$.
However, we will see that the number of bits needed is not a function of $\alpha - \beta$, and that guarantees can be improved when these parameters have values close to~1.

Let $x$ and $y$ be vectors from the $d$-dimensional Euclidean unit sphere $\S^{d-1}$.
When represented in a computer with limited precision floating or fixed-point numbers, the precision we can obtain when computing the inner product $\doot x y$ will of course depend on the precision of the representation of $x$ and $y$.
Suppose we round coordinates $x$ and $y$ to the nearest integer multiple of $\varepsilon / d$, for some parameter $\varepsilon > 0$, to produce ``uniformly quantized'' vectors $x'$ and $y'$.
Then it is easy to see that the difference between $\doot x y$ and $\doot{x'}{y'}$ is at most~$\varepsilon$.
The space required to store each coordinate $x'_i, y'_i \in [-1,+1]$ is $\lceil \log_2(2d/\varepsilon)\rceil$ bits, so we get $d\log_2(d/\varepsilon) + \BOx{d}$ bits in total using standard, uniform quantization.
 On the upper bound side we know that the number of bits per dimension can be made independent of~$d$.
 The following lemma appears to be folklore --- a proof can be found in~\cite[Theorem 4.1]{Alon2017optimal}. 
\begin{lemma}\label{thm:main-upper}    For every $\varepsilon > 0$ there exists a mapping $\mathcal{E}: S^{d-1} \rightarrow \{0,1\}^\ell$, where $\ell = d\log_2(1/\varepsilon) + \BOx{d}$ such that $\doot x y$ can be estimated from $\enc{x}$ and $\enc{y}$ with additive error at most $\varepsilon$.
 \end{lemma}

 In this paper we ask if this space usage is optimal for the problem of distinguishing between two specific inner product values.
 Our methods will work through a decoding function that produces a unit vector from a bit representation (i.e., the approximation is a result of quantizing the input vectors).
 Specifically, we consider the following problem:

   \begin{definition}		 For positive integers $d$ and $\ell$, and $\alpha,\beta\in [0,1]$ with $\alpha>\beta$ the $(\alpha,\beta,d,\ell)$-{\sc InnerProduct} problem is to construct mappings $\mathcal{E}: \S^{d-1}\rightarrow \{0,1\}^{\ell}$ and $\mathcal{D}: \{0,1\}^\ell\rightarrow \S^{d-1}$ and choose $t\in\R$, such that for every choice of unit vectors $x,y\in \S^{d-1}$ we have:
  \begin{align*}
    \doot xy\geq\alpha &\implies \doot {\dec{\enc{x}}} {\dec{\enc{y}}}\geq t\text{ and }\\
    \doot xy\leq\beta &\implies \doot {\dec{\enc{x}}} {\dec{\enc{y}}} < t \enspace .
  \end{align*}
  We refer to the parameter $\ell$ as the \emph{space usage} of a construction.
	Whenever $d$ and $\ell$ are understood from the context we omit them and talk about the $(\alpha,\beta)$-{\sc InnerProduct} problem.
\end{definition}

On the upper bound side our main technical lemma is the following:
\begin{theorem} \label{thm:distinguish}
  $(\alpha,\beta)$-{\sc InnerProduct} can be solved using space $\ell=d\log_2\left(\frac {\sqrt{1-\beta}}{\alpha-\beta}\right)+\BOx{d}$.
\end{theorem}

Theorem~\ref{thm:distinguish} upper bounds the space needed to approximate inner products between unit vectors.
For example we can distinguish pairs with inner product $\alpha=1-\epsilon$ from pairs with inner product less than $\beta=1-2\epsilon $ using space $\frac{d}{2} \log_2{\frac 1 \epsilon }+\BOx{d}$.
The problem is closely linked to estimation, so it is unsurprising that it matches the bound in Lemma~\ref{thm:main-upper} for $\alpha-\beta=\epsilon$ in the worst case of $\beta=0$. What is interesting is that for $\beta$ close to $1$ we get improved constants, saving up to a factor 2 on the space when $\alpha$ approaches~1.

Our proof uses a variant of \emph{pyramid vector quantization}~\cite{fischer1986pyramid} and the technique is essentially an implementation of a grid-based $\varepsilon$-net as described in~\cite{Alon2017optimal}, though the analysis is different.
The exposition is supposed to be self-contained, and in particular we do not assume that the reader is familiar with pyramid vector quantization or $\varepsilon$-nets.

\medskip

Finally, we show a tight lower bound.
Consider a communication protocol where Alice is given $x\in \S^{d-1}$ and Bob is given $y\in \S^{d-1}$.
For parameters $\alpha, \beta \in (0,1)$, with $\alpha = \beta + \varepsilon$, known to both parties, how many bits of information does Alice need to send to Bob in order for Bob to be able to distinguish the cases $\langle x,y\rangle \geq \alpha$ and $\langle x,y\rangle \leq \beta$?
Specifically, how many bits must Alice send, in the worst case over all vectors $x$, to allow Bob to answer correctly for every vector $y$?
We note that a solution for the $(\alpha,\beta)$-{\sc InnerProduct} problem implies a communication protocol using $\ell$ bits, but our lower bound applies more generally to any one-way communication protocol, not necessarily based on quantization.

\begin{theorem}\label{thm:main-lower}   For each choice of $\alpha,\beta\in (0,1)$ with $\alpha > \beta$, suppose that there exists a mapping $\mathcal{E}: S^{d-1} \rightarrow \{0,1\}^\ell$ such that for all $x,y\in S^{d-1}$ we can determine from $\enc{x}$ and $\enc{y}$ whether $\doot x y \leq \beta$ or $\doot x y \geq \alpha$ (or output anything if $\doot x y \in (\beta,\alpha))$.   
   Then $\ell \geq d \log_2 \left(\tfrac{\sqrt{1-\beta}}{\alpha - \beta}\right) - \BOx{d}$.
\end{theorem}
This matches the upper bound up to the additive term of $\BOx{d}$ bits.

\section{Further related work}\label{sec:related}

\subparagraph*{Motivating applications.}
Calculating the inner product of two vectors is a frequently used sub-routine in linear algebra,
and many machine learning algorithms rely heavily on inner product calculation.
For example, the inner loop of algorithms for training of complex neural networks uses millions and millions of inner product computations.
Often what is ultimately learned is an embedding onto a high dimensional unit sphere where the inner product can be used directly as a similarity measure.

In such large scale computations the bottleneck is often the limited bandwidth of the hardware in question, and having slightly smaller vector representations can massively improve the execution time.
This gives rise to the idea of computing inner products with reduced precision.
Recently, several studies showed that deep neural networks can be trained using low precision arithmetic, see e.g.~\cite{chen2015compressing,gupta2015deep,han2015deep}.
This has led to a new generation of software and reduced-precision hardware for machine learning algorithms:
\begin{itemize}
	\item NVIDIA's TensorRT GPU framework and Google's TensorFlow and Tensor Processing Unit, that both operate with 8- or 16-bit fixed point number representations, and 
	\item Intel's Nervana processor that uses the so-called \emph{FlexPoint} vector representation~\cite{flexpoint}, combining 16-bit uniform quantization with a shared exponent that allows representing vectors in a large dynamic range of magnitudes.
\end{itemize}
From a theoretical point of view these hardware representations use at least $\log_2 d \pm O(1)$ more bits per dimension than what is required to ensure a given additive error.

Reduced precision inner products have also been employed in knowledge discovery~\cite{blalock2017bolt} and similarity search~\cite{guo2016quantization,jegou2011product}.

\subparagraph*{Dimensionality reduction.}
There is a large literature on the space complexity of estimating Euclidean distances, usually studied in the setting where a certain failure probability $\delta > 0$ is allowed, and with number of \emph{dimensions} (rather than bits) as the measure of space.
For certain ``random projection'' mappings $f: \R^d\rightarrow \R^D$ one can estimate the Euclidean distance $||x-y||_2$ from $f(x)$ and $f(y)$ up to a multiplicative error of $1+\varepsilon$, with failure probability~$\delta$.
It is known that using $D=\BOx{\log(\delta^{-1})\varepsilon^{-2}}$ dimensions is necessary~\cite{jayram2013optimal,kane2011almost} and sufficient~\cite{johnson1984extensions}.
For unit vectors this implies an approximation of inner products with $\BOx{\varepsilon}$ additive error through the identity 
\begin{equation}\label{eq:length2inner}
\doot x y = \tfrac{1}{2}(||x||_2^2+||y||_2^2-||x-y||_2^2) \enspace .
\end{equation}
Using (a specific type of) random projections to estimate inner products, with an additive error guarantee, is known as ``feature hashing''~\cite{weinberger2009feature}.

Indyk et al.~\cite{indyk2017practical, indyk2017near} considered the bit complexity of representing all distances, up to a given relative error $1+\varepsilon$, within a given set $S$ of $n$ vectors in $\R^d$.
For this problem one can assume without loss of generality that $d=\BOx{\varepsilon^{-2}\log n}$, using dimension reduction.
Suppose that we only need to preserve distances of unit vectors up to an additive $\varepsilon$, which implies that inner products are preserved up to $\BOx{\varepsilon}$.
Then for $d = \varepsilon^{-2}\log n$ the space usage per point of the method described in~\cite{indyk2017near} is $\BOx{d\log(1/\varepsilon)}$.
This is within a constant factor of our upper bound, but not directly comparable to our result which works for all unit vectors.
Recently, Indyk and Wagner~\cite{indyk2018approximate} studied the space required to solve the $d$-dimensional Euclidean $(1+\varepsilon)$-approximate nearest neighbor problem in the setting where vector coordinates are integers in a bounded range (e.g.~of size $n^{\BOx{1}}$).
While this method gives guarantees for new vectors outside of $S$ their method is randomized and can fail to correctly determine an approximate nearest neighbor, while our method is deterministic.

\subparagraph*{Vector quantization.}

In a nutshell, vector quantization~\cite{gersho2012vector} is the process of mapping vectors in a space (usually a bounded subset of Euclidean space) to the nearest in a finite set of vectors~$Q$.
The goal is to minimize the size of $Q$ and the distance between vectors and their quantized versions, often with respect to a certain distribution of source (or input) vectors.
Fischer first described \emph{pyramid vector quantization}~\cite{fischer1986pyramid}, showing that it is near-optimal for Laplacian sources.
Since high-dimensional Laplacian vectors have lengths that are tightly concentrated around the expectation, it is natural to speculate if the method is also near-optimal for fixed-length (or unit) vectors.
It turns out to be easier to analyze a variant of pyramid vector quantization for which we can show that this is indeed the case.
This is described in section~\ref{sec:upper}.

Quantization methods have previously been used to speed up nearest neighbor search.
The technique of \emph{product quantization}~\cite{jegou2011product} has been particularly successful for this application.
Product quantization uses an initial random rotation of input vectors followed by application of an optimal quantization method on low-dimensional blocks.
Since the size of the codebook is fixed for each block the resulting quantization error cannot be bounded with probability 1.

Quantization of the unit sphere has been studied in complexity theory as 
 \emph{$\varepsilon$-nets for spherical caps}.
 Rabani and Shpilka~\cite{rabani2010explicit} give a construction in which $|Q|$ is polynomially related to the best size possible with a given quantization error.
 Along and Klartag~\cite{Alon2017optimal} use such nets to achieve $|Q|$ that is within a factor $\exp(\BOx{d})$ of optimal, improving~\cite{rabani2010explicit} whenever the quantization error is $o(1)$, such that $|Q|$ must be superexponential in~$d$.

In the literature on machine learning (and its application areas) there is a myriad of methods for learning a data-dependent quantization mapping that exploits structure in a data set to decrease the quantization error.
We refer to the survey of Wang et al.~\cite{wang2016learning} for details.
In contrast to such methods, we seek guarantees for all vectors, not just vectors from a given data set or distribution.

\subparagraph*{Communication complexity.}

Consider a communication protocol in which Alice and Bob are given unit vectors $x,y\in \left\{\pm \tfrac{1}{\sqrt{d}}\right\}^d$ and need to approximate $\doot x y$.
The \emph{gap hamming} problem is the special case where the task is to distinguish between cases of weak positive and negative correlation: 
Is $\doot x y > 1/\sqrt{d}$ or is $\doot x y < - 1/\sqrt{d}$?
This problem is known to require $\Omega(d)$ bits of communication~\cite{chakrabarti2012optimal,indyk2003tight,sherstov2012communication}, even for randomized protocols with error probability, say, $1/3$.
In turn, this implies a lower bound for the space complexity of estimating inner products, since a space complexity of $\ell$ bits implies a (one-way) communication protocol using $\ell$ bits of communication.
The lower bound extends to arbitrary thresholds $\alpha$ and $\beta$ with $\alpha - \beta = \Theta(1/\sqrt{d})$ by translation.
In this paper we consider general unit vectors $x,y\in\R^d$ and are able to show a higher lower bound of $\tfrac{1}{2} d \log_2 d - \BOx{d}$ bits for distinguishing inner products of distance $\varepsilon = \Theta(1/\sqrt{d})$.

\section{Upper bound}\label{sec:upper}
We use a well-known grid-based rounding method to construct our representation, see e.g.~\cite{Alon2017optimal,fischer1986pyramid}.
For completeness we provide a simple, self-contained description of a representation and show that it has the properties described in Theorem~\ref{thm:distinguish}.
The grid resolution is controlled by a parameter $\delta\in [0,1]$.
For every vector $x\in\R^d$ let $f(x)\in \R^d$ be defined by
  \begin{equation}\label{eq:pyramid}
  f(x)=x'/||x'||_2, \text{ where } x_i'=\left\lfloor x_i \tfrac{\sqrt{d} } \delta + \tfrac 1 2 \right\rfloor \frac \delta {\sqrt{d}} \text{ for } i=1,\dots,d \enspace .
  \end{equation}

  It is clear that the number of bits for storing a single integer coordinate $x_i'$ can be large in high dimensions, up to $\log_2(2\sqrt{d}/\delta)$ bits, but we can give a much better bound on the \emph{average} number of bits per coordinate.
  If $\|x\|\leq1$ we can store $f(x)$ using $\ell=d\log_2(1 / \delta)+\BOx{d}$ bits of space.
  To compute $x'$ it suffices to know the integers $z_i=\lfloor x_i \tfrac{\sqrt{d} } \delta + \tfrac 1 2 \rfloor$.
  We first allocate $d$ bits to store the set $\{ i \; | \; z_i < 0\}$, such that it only remains to store the sequence of absolute values $|z_1|,\dots,|z_d|$.
     Next, using $\|x\|_2\leq 1$ we observe that 
	 	 \[\sum_{i=1}^d |z_i| \leq \|x\|_1 \tfrac{\sqrt{d}}{\delta} + d/2 \leq \sqrt{d}\,\|x\|_2\,\tfrac{\sqrt{d}}{\delta} + d/2 \leq d / \delta + d/2 \enspace .\]
	      Thus if we set $s=\left\lfloor d/\delta + d/2\right\rfloor$ we can encode $|z_1|,\dots,|z_d|$ by specifying a partitioning of $s$ elements into $d+1$ parts.
     There are ${s+d\choose d}$ such partitionings so we can assign each vector a unique representation of $\ell = \left\lceil \log_2{s+d\choose d}\right\rceil + d
	      = d\log_2\left(1/\delta \right) + \BOx{d}$ bits.

Before proving Theorem~\ref{thm:distinguish} we give a simple space bound for distance distortion which is useful in it own right.

\begin{lemma}
  \label{lem:distortion}
   For $\delta \leq 1$ and every choice of $x,y\in\R^d$, defining $f$ according to (\ref{eq:pyramid}) we have:
   \[\| x-y \|_2- \delta \leq \|f(x)-f(y)\|_2 \leq \|x-y\|_2 + \delta\enspace .\]
\end{lemma}

   \begin{proof}

     Observe that $|x_i-x_i'|\leq \frac \delta {2\sqrt{d}}$. This means that:
          \[ \|x-x'\|_2 = \sqrt{\sum_{i=1}^d (x_i-x_i')^2}\leq \sqrt{\sum_{i=1}^d \frac {\delta^2} {4d}} = \frac \delta 2 \enspace . \]
	 	 Since $\delta \leq 1$ the angle between $x$ and $x'$ is bounded by $\pi/3$, and hence $\|x-f(x)\|_2 \leq 1$.
	 This implies that
	 $\|x-f(x)\|^2_2 = 2 - 2\doot{x}{x'/||x'||_2} \leq 1+\|x'\|_2 - 2\doot{x}{x'} =\|x-x'\|^2_2$, and in particular we get $\|x-f(x)\|_2 \leq \|x-x'\|_2 \leq \delta / 2$.
     Finally, using the triangle inequality:
	 	 \begin{align*}
	 \|f(x)-f(y)\|_2 & \leq \|x-f(x)\|_2 + \| x - y \|_2 + \| y-f(y)\|_2 \leq \|x-y\|_2 + \delta  ,\text{ and}\\
     \|f(x)-f(y)\|_2 & \geq\|x-y\|_2-\|x-f(x)\|_2-\|y-f(y)\|_2\geq \|x-y\|_2-\delta \enspace .
	 \end{align*}
\end{proof}
We are now ready to prove Theorem~\ref{thm:distinguish}.
\begin{proof}[Proof of Theorem~\ref{thm:distinguish}]
	Let the encoding function $\enc{\cdot}$ map a vector $x$ to the $\ell$-bit representation of $x'$ as defined in $(\ref{eq:pyramid})$.
    The decoding function $\dec{\cdot}$ is defined such that $\dec{\enc{x}} = f(x)$.

    By Lemma~\ref{lem:distortion} we have:
  \[\max\{0,\| x- y \|- \delta\} \leq \|f(x)-f(y)\| \leq \|x- y\| + \delta\enspace .\]

\noindent
  Using the distance bounds and the identity (\ref{eq:length2inner}) several times we get:
    \begin{align*}
    \doot {f(x)}{f(y)} &= \tfrac 12 (\|f(x)\|^2+\|f(y)\|^2-\|f(x)-f(y)\|^2)\\
                        &\leq\min\{1,\tfrac12 (2+2\|x-y\|\delta-\|x-y\|^2-\delta^2)\}\\
                        &=\min\{1,\doot x y +\|x-y\|\delta-\delta^2/2\}, \text{ and }\\[2ex]
    \doot {f(x)}{f(y)} &\geq\tfrac12 (2-2\|x-y\|\delta-\|x-y\|^2-\delta^2)\\
                        &=\doot x y -\|x-y\|\delta-\delta^2/2 \enspace .
  \end{align*}

  \noindent
  We can then see
    \begin{align*}
      \doot x y\geq \alpha \implies \doot {\dec{\enc{x}}} {\dec{\enc{y}}}&\geq \alpha-\delta\sqrt{2-2\alpha}-\delta^2/2\text{ and}\\
      \doot x y\leq \beta \implies \doot {\dec{\enc{x}}} {\dec{\enc{y}}}&\leq \min\{1,\beta+\delta\sqrt{2-2\beta}-\delta^2/2\}
    \end{align*}
    Setting $\delta=\frac {\alpha-\beta}{2\sqrt{2-2\beta}}$ and $t=\alpha-\delta\sqrt{2-2\alpha}-\delta^2/2$ we get a valid solution to the $(\alpha,\beta)$-{\sc InnerProduct} problem.
\end{proof}

The same grid, as specified by $\delta$, works for every $(\alpha',\beta')$-{\sc InnerProduct} instance (with a suitable choice of threshold $t$) as long as 
		\begin{equation}\label{eq:delta-bound}
	\delta < \frac{\alpha'-\beta'}{\sqrt{2-2\alpha'}+\sqrt{2-2\beta'}} \enspace .
	\end{equation}

Note that Lemma~\ref{thm:main-upper} also follows from Theorem~\ref{thm:distinguish}:
For a given inner product value $p$ consider $\alpha' = p + \varepsilon / 2$ and $\beta' = p - \varepsilon / 2$.
Setting $\delta = \varepsilon / 4$, we satisfy (\ref{eq:delta-bound}) for all $p$ and get that for every choice of unit vectors $x,y\in \S^{d-1}$:
  \[\doot x y - \epsilon \leq \doot {\dec{\enc{x}}}{\dec{\enc{x}}}\leq \doot x y + \epsilon  \enspace. \]

\section{Lower bound}\label{sec:lower}

Consider a one-way communication protocol solving the $(\alpha,\beta)$-{\sc InnerProduct} problem where Alice sends a string $\enc{x}$ to Bob, and Bob must be able to output a real number $p(\enc{x},y)$ and a threshold $t\in\R$ such that 
$$\langle x,y\rangle \geq \alpha \implies p(\enc{x},y)\geq t \text{, and}$$ 
$$\langle x,y\rangle \leq \beta\implies p(\enc{x},y) < t \enspace . $$
Note that there is no requirement on $p(\enc{x},y)$ whenever $\langle x,y\rangle \in (\beta,\alpha)$.
We wish to answer the following question: How many bits must Alice send, in the worst case over all vectors $x$, to allow Bob to answer correctly for every vector $y$?

Let $d>1$ be an integer.
For $\varepsilon > 0$ and $x\in \R^d$ let 
$B_\varepsilon(x) = \{ y\in \R^d \; | \; ||x-y||_ 2\leq\varepsilon \}$ 
be the ball of radius $\varepsilon$ centered at $x$, let $B_1 = B_1(0)$ be the unit ball centered at the origin, and denote by $\text{cap}_{\Theta}(x) = \{ y \in \S^{d-1} \; | \; \langle x, y\rangle \geq \cos \Theta \}$ the unit spherical cap around $x$ with polar angle $\Theta$.

\subsection{Preliminaries}

The gamma function is an extension of the factorial function to complex numbers.
We will need the following formulas for the value of the gamma function on integers and half-integers (see e.g.~\cite{vidunas2005}). For integer $n > 0$:
\begin{equation}\label{eq:gamma}
\begin{split}
\Gamma(n + \tfrac{1}{2}) & = \frac{(2n)!\sqrt{\pi}}{2^{2n} n!}\\
\Gamma(n + 1) & = n!
\end{split}
\end{equation}

\begin{lemma}\label{lem:gamma_fraction}
	For integer $d>1$, $\Gamma(d/2+\tfrac{1}{2}) / \Gamma(d/2+1) > 1/(3\sqrt{d})$.
\end{lemma}
\begin{proof}
	We use Stirling's approximation to the factorial:
		$$ \sqrt{2\pi} n^{n+1/2} e^{-n} \leq n! \leq e\, n^{n+1/2} e^{-n}$$
		Together with (\ref{eq:gamma}) we get, for even $d$:
		\begin{alignat*}{1}
		\Gamma(d/2+\tfrac{1}{2}) / \Gamma(d/2+1)  
		= \frac{d!\sqrt{\pi}}{2^d((d/2)!)^2}
		\geq \frac{\sqrt{2}\pi \, d^{d+1/2} e^{-d}}{2^d e^2 (d/2)^{d+1} e^{-d}}
		 = \frac{2\sqrt{2}\pi}{e^2 \sqrt{d}} 
		 > 1/(3\sqrt{d}) \enspace .
	\end{alignat*}
	The case of odd $d$ is similar.
\end{proof}

\begin{lemma}\label{lem:vol}
	Let $c_d = \frac{\pi^{d/2}}{\Gamma(d/2+1)}$, $r>0$, and $\delta\in(0,1)$. Then:
\begin{enumerate}
\item\label{eq:1} The volume of $B^d_r$, the $d$-dimensional ball of radius $r$, is 
$c_d\, r^d$.
\item\label{eq:2} The surface area of $\S^{d-1}_r$, the $d$-dimensional sphere of radius $r$, is
$c_{d}\, d\, r^{d-1}$.
\item\label{eq:3} The surface area of $\text{cap}^d_{\Theta}(x)$, a unit spherical cap with polar angle $\Theta$, is at most $$c_{d-1}\, d\, (2(1 - \cos \Theta))^{(d-1)/2} \enspace .$$
\end{enumerate}
\end{lemma}
\begin{proof}
	Volume bound \ref{eq:1}.~is standard, see e.g.~\cite[page 11]{pisier1999volume}.
	Differentiating with respect to~$r$ gives the surface area in line~\ref{eq:2}.
	For the upper bound \ref{eq:3}.~we express the surface area as an integral.
	Let $r(h)=\sqrt{h(2-h)} < \sqrt{2h}$ be the radius of the $d-1$-dimensional sphere at the base of the cap of height $h$.
	Note that the sphere has surface area $c_{d-1}\, d\, r(x)^{d-1}$.
	Integrating over cap heights from $0$ to $1-\cos\Theta$ we bound the surface area:
		$$\int^{1 - \cos \Theta}_0 c_{d-1}\, d\, r(x)^{d-1} dx \leq (1 - \cos \Theta) c_{d-1}\, d\, (2(1 - \cos \Theta))^{(d-1)/2} \enspace . $$
		The inequality uses that $r(x)\leq r(1 - \cos \Theta) \leq \sqrt{2(1 - \cos \Theta)}$ for $x\in [0,1 - \cos \Theta]$.
\end{proof}

\begin{lemma}
	For every $\Theta \in (0,\pi/2)$ there exists a code $\mathcal{C}_\Theta \subset S^{d-1}$ of size 
		\[|\mathcal{C}_\Theta| \geq (2(1-\cos\Theta))^{(d-1)/2} / (3\sqrt{d}) \] 
		such that for all $x,y\in\mathcal{C}_\Theta$ with $x\ne y$ we have $\langle x, y\rangle \geq \cos\Theta$.
\end{lemma}
\begin{proof}
	We follow the outline of the standard non-constructive proof of the Gilbert-Varshamov bound.
	That is, we argue that $\mathcal{C}_\Theta$ can be constructed in a greedy manner by adding an additional point from 
		$$\S^{d-1} \backslash \bigcup_{x\in\mathcal{C}_\Theta} \text{cap}^d_\Theta(x)$$ 
		until this set is empty.
	Clearly this construction produces a set $\mathcal{C}_\Theta$ with the property that every pair of points have angle greater than $\Theta$.
	We observe that the procedure can stop only when the area of $\bigcup_{x\in\mathcal{C}_\Theta} \text{cap}_{\Theta}(x)$ exceeds that of $\S^{d-1}$.
	The number of iterations, and thus the size of $\mathcal{C}_\Theta$ is at least the ratio between the surface area of the unit sphere and a spherical cap, $\text{cap}_{\Theta}(\cdot)$.
	In turn, this is lower bounded by the ratio of bound number~\ref{eq:2}~(with $r=1$) and~\ref{eq:3}~from Lemma~\ref{lem:vol}. Using Lemma~\ref{lem:gamma_fraction} we get:
		$$ \frac{c_d\, d}{c_{d-1}\, d\, (2(1-\cos\Theta))^{(d-1)/2}} \geq (2(1-\cos\Theta))^{-(d-1)/2} / (3\sqrt{d}) \enspace .$$
\end{proof}

\subsection{Space complexity}

Define $\Theta = \arccos \beta - \arccos \alpha$. 
We claim that for every pair of vectors $x_1,x_2\in \S^{d-1}$ with angle $\theta = \arccos \, \langle x_1, x_2\rangle \geq \Theta$ there exists a vector $y\in \S^{d-1}$ such that $\langle x_1, y\rangle =\beta$ and $\langle x_2, y\rangle \geq \alpha$.
Specifically, let
\begin{equation}\label{eq:y-def}
y = y(x_1,x_2)
= \left(\beta - \sqrt{1-\beta^2}\, \tfrac{\cos\theta}{\sin\theta}\right) x_1 + \tfrac{\sqrt{1-\beta^2}}{\sin\theta} \, x_2 \enspace .
\end{equation}
To see that $y$ is indeed a unit vector we compute $||y||_2^2 = \langle y,y\rangle$.
Since $\langle x_1, x_2\rangle = \cos\theta$ (by definition of $\theta$) and $\langle x_1, x_1\rangle = \langle x_2, x_2\rangle = 1$ (since $x_1$ and $x_2$ are unit vectors) we get:
\begin{align*}
\langle y,y\rangle & = \left(\beta - \sqrt{1-\beta^2} \, \tfrac{\cos\theta}{\sin\theta}\right)^2 + \left(\tfrac{\sqrt{1-\beta^2}}{\sin\theta}\right)^2 + 2 \left(\beta - \sqrt{1-\beta^2}\, \tfrac{\cos\theta}{\sin\theta}\right) \left(\tfrac{\sqrt{1-\beta^2}}{\sin\theta}\right) \cos\theta\\
 & = \beta^2 + \tfrac{1-\beta^2}{\sin^2\theta} (1-\cos^2\theta)\\
 & = \beta^2 + (1-\beta^2) = 1 \enspace .
\end{align*}
The third equality uses the Pythagorean identity $\cos^2\theta + \sin^2\theta = 1$.
Next, we check that $y$ has the claimed inner products with $x_1$ and $x_2$:
\begin{align*}
\langle x_1, y \rangle & = \left(\beta - \sqrt{1-\beta^2}\, \tfrac{\cos\theta}{\sin\theta}\right) + \sqrt{1-\beta^2}\, \tfrac{\cos\theta}{\sin\theta} = \beta,\\
\langle x_2, y \rangle 
& = \left(\beta - \sqrt{1-\beta^2}\, \tfrac{\cos\theta}{\sin\theta}\right) \cos\theta + \tfrac{\sqrt{1-\beta^2}}{\sin\theta}\\
& = \langle (\beta, \sqrt{1-\beta^2}), (\cos\theta,\tfrac{1-\cos^2\theta}{\sin\theta})\rangle\\
& = \langle (\beta, \sqrt{1-\beta^2}), (\cos\theta,\sin\theta)\rangle
\geq \alpha \enspace .
\end{align*}
The final inequality follows since the angle between the vectors $(\beta, \sqrt{1-\beta^2})$ and $(\cos\theta,\sin\theta)$ is at most $\arccos(\beta) - \theta \leq \arccos(\beta) - \Theta = \arccos\alpha$.

\medskip

Now consider the code $\mathcal{C}_\Theta$.
For distinct $x_1,x_2\in \mathcal{C}_\Theta$ we must have, for $y = y(x_1,x_2)$ as defined in (\ref{eq:y-def}), $p(e(x_1),y) < t \leq p(e(x_2),y)$, which means that $e(x_1)\ne e(x_2)$.
Hence $R = \{ \enc{x} \; | \; x\in \S^{d-1}\}$ must contain at least $|\mathcal{C}_\Theta|$ binary strings, and in particular 
\begin{align*}
	\ell \geq \log_2 |\mathcal{C}_\Theta| 
	& \geq \log_2\left( (2(1-\cos\Theta))^{-(d-1)/2} / (3\sqrt{d}) \right)\\
	& \geq \tfrac{d}{2} \log_2\left(\frac{1}{1-\cos\Theta}\right) - \BO{d}\\
	& \geq d \log_2 (1/\Theta) - \BO{d} \enspace .
\end{align*}
For the final inequality we use that $1-\cos\Theta \leq \Theta^2$, which holds when $\Theta \in (0,\pi/2)$.
To finish the proof we will show that whenever $0 \leq \beta \leq \alpha\leq 1$:
\begin{equation}\label{eq:arccos}
\Theta = \arccos\beta - \arccos\alpha \leq \tfrac{\pi}{2} \tfrac{\alpha-\beta}{\sqrt{1-\beta}} \enspace .
\end{equation}
For each fixed $\beta\in [0,1]$ we must show that $\arccos \alpha \geq \arccos\beta - \tfrac{\pi}{2} \tfrac{\alpha-\beta}{\sqrt{1-\beta}}$ for all $\alpha\in [\beta,1]$.
Since $\alpha \mapsto \arccos\alpha$ is concave for $\alpha\in [0,1]$, and the function $\alpha\mapsto \arccos\beta - \tfrac{\pi}{2} \tfrac{\alpha-\beta}{\sqrt{1-\beta}}$ is linear, it suffices to check the inequality at the endpoints where $\alpha=\beta$ and $\alpha = 1$.
In the former case we clearly get equality.
In the latter we use the fact $\arccos \beta - \tfrac{\pi}{2} \sqrt{1-\beta} \leq 0$ for $\beta\in [0,1]$ 
to see that the inequality (\ref{eq:arccos}) holds.

This proves that $\ell \geq d\log_2\left(\frac{\sqrt{1-\beta}}{\alpha-\beta}\right) - \BOx{d}$ bits are needed, establishing Theorem~\ref{thm:main-lower}.

\section{Experiments}

Since the encoding in our upper bound is potentially practical, we evaluated the accuracy of the method experimentally on several data sets.
We also computed the space usage of each set of vectors, based on an optimal encoding of the method of Section~\ref{sec:upper} with various values of parameter $\delta$, and compared it to a baseline of using 32-bit floating point numbers.

We considered three data sets, MNIST (handwritten digits), SIFT (image features), and DLIB-FACES (neural net embeddings of faces on a 128-dimensional unit sphere).
The data sets are chosen to span different distributions of entry magnitude.
Whereas MNIST vectors are approximately sparse, DLIB vectors have smoothly distributed magnitudes.
MNIST and SIFT are not natively unit vectors, so we normalized the vectors prior to encoding.

We computed the inner product error for 2000 vector pairs in each data set.
Table~\ref{experiments} shows the maximum absolute error observed when calculating inner products using the decoded vectors compared to using the original vectors.
It also shows the median absolute error and the error at the 90th percentile.

In all cases the observed errors are well below the worst case bound of~$\epsilon = 4\delta$.
This can partly be explained by sparsity of vectors, since our method represents zero entries in vectors with no error.
Also, while the effect of $d$ rounding errors is $d$ times the individual error in the worst case (which is what our theoretical results bound), in the typical case errors will tend to cancel since not all errors go in the same direction.
Heuristically we could imagine that errors are independent and random, in which case we would expect the sum of all errors to be proportional to $\sqrt{d}$ rather than $d$.

\begin{table}[t]
\centering

\begin{tabular}{lll|c|lll}
 Dataset & $d$      &  $\delta$  & space 	& median  & 90th pct.       & max                        \\\hline
  MNIST  & 784      &  0.1         & 16\% 	& 0.0019  &  0.0052         & 0.0165              \\
  SIFT   & 128      &  0.1         & 16\% 	& 0.0034  &  0.0084         & 0.0181               \\
  DLIB   & 128      &  0.1         & 16\% 	& 0.0028  &  0.0071         & 0.0135              \\\hline
  MNIST  & 784      &  0.05        & 19\% 	& 0.0008  &  0.0027         & 0.0103               \\
  SIFT   & 128      &  0.05        & 19\% 	& 0.0019  &  0.0044         & 0.0091               \\
  DLIB   & 128      &  0.05        & 19\% 	& 0.0012  &  0.0033         & 0.0075              \\\hline
  MNIST  & 784      &  0.01        & 26\% 	& 0.0002  &  0.0007         & 0.0027              \\
  SIFT   & 128      &  0.01        & 26\% 	& 0.0004  &  0.0008         & 0.0100               \\
  DLIB   & 128      &  0.01        & 26\% 	& 0.0002  &  0.0005         & 0.0011             
\end{tabular}
\caption{\small Space usage and absolute inner product error using the method of Section~\ref{sec:upper} on pairs of vectors from various real-life data sets.
The space usage is measured against a baseline of using a 32-bit floating point numbers to represent each of the $d$ dimensions.
Observed errors are smaller than our worst-case bound of $4\delta$, probably due to cancellation effects.}\label{experiments}
\end{table}

\section{Conclusion}

We have established tight upper and lower bounds for the problem of representing unit vectors such that inner products can be estimated within a given additive error (with probability~1).
An interesting possibility would be to consider relative error estimates of Euclidean distances (as in the recent work~\cite{indyk2018approximate}) while not allowing any failure probability. Another potential direction would be to achieve provably smaller \emph{expected} error, while preserving the worst-case guarantees, by applying an initial random rotation to all data vectors.

\bibliography{bibliography}
\end{document}